%% file: FullBudgetedGraphColoring.tex
\documentclass[11pt]{article}
\usepackage[margin=1in]{geometry}
\usepackage{graphicx}
\usepackage{authblk}
\usepackage{tcolorbox}
\usepackage[ruled,vlined]{algorithm2e}
\usepackage{amsfonts}
\usepackage{mathtools}
\usepackage{xspace}
\usepackage{hyperref}
\usepackage{tcolorbox}

\DeclareRobustCommand{\rgamma}{{\mathpalette\irgamma\relax}}
\newcommand{\irgamma}[2]{\raisebox{\depth}{$#1\gamma$}}

\usepackage{amsmath}
\usepackage{amsthm}
\usepackage{comment}
\usepackage{xcolor}

\usepackage{array}
\newcolumntype{L}[1]{>{\raggedright\let\newline\\\arraybackslash\hspace{0pt}}m{#1}}
\newcolumntype{C}[1]{>{\centering\let\newline\\\arraybackslash\hspace{0pt}}m{#1}}
\newcolumntype{R}[1]{>{\raggedleft\let\newline\\\arraybackslash\hspace{0pt}}m{#1}}

 \DeclarePairedDelimiter\ceil{\lceil}{\rceil}
\DeclarePairedDelimiter\floor{\lfloor}{\rfloor}
 \newtheorem{theorem}{Theorem}
 \newtheorem{lemma}{Lemma}
 \newtheorem{claim}{Claim}

 \newtheorem{corollary}{Corollary}
\newcommand{\bcp}{{\sc BCP}\xspace}
\newcommand{\bocp}{{\sc BoCP}\xspace}
\newcommand{\ecp}{{\sc ECP}\xspace}
\newcommand{\dsp}{{\sc Dominating Set Problem}\xspace}

\newcommand{\partp}{{\sc $3$-Partition Problem}\xspace}
\newcommand{\ebcp}{{\sc EBCP}\xspace}
\newcommand{\I}{\mathcal{I}}
\newcommand{\B}{\mathcal{B}}
\newcommand{\K}{\mathcal{K}}

\newcommand{\Ge}{\mathcal{G}}
\newcommand{\Pe}{\mathcal{P}}

\newcommand{\fpt}{{\sf FPT}\xspace}
\newcommand{\npc}{\textsf{NP}-complete\xspace}
\newcommand{\nph}{\textsf{NP}-hard\xspace}
\newcommand{\woh}{\textsf{W[1]}-hard\xspace}
\newcommand{\wth}{\textsf{W[2]}-hard\xspace}
\newcommand{\wih}{\textsf{W[i]}-hard\xspace}

\newcommand{\defproblem}[3]{
	\begin{tcolorbox}[colback=gray!5!white,colframe=gray!75!black]
		\begin{tabular*}{\textwidth}{@{\extracolsep{\fill}}lr} #1   \\ \end{tabular*}
		{\bf{Input:}} #2  \\
		{\bf{Question:}} #3
	\end{tcolorbox}
}

\begin{document}
\title{Structural Parameterizations of Budgeted Graph Coloring}
%
%
%
%
\author[1]{Susobhan Bandopadhyay}
\author[2]{Suman Banerjee}
\author[1]{Aritra Banik}
\author[3]{Venkatesh Raman}
\affil[1]{
National Institute of Science Education and Research, HBNI, Bhubaneswar, India.

E-mail: \{susobhan.bandopadhyay,aritra\}@niser.ac.in} 
\affil[2]{Department of Computer Science and Engineering, Indian Institute of Technology Jammu, India. 

E-mail: suman.banerjee@iitjammu.ac.in}
\affil[3]{The Institute of Mathematical Sciences, HBNI, Chennai, India.

E-mail:vraman@imsc.res.in}

\maketitle              
\begin{abstract}
\input{abstract_mod}

\end{abstract}
%
%
%
\input{intro}

\input{NotationsandPreliminaries}

\input{special}

\input{fpt}


\input{exact}
\input{Conclusion}

\bibliographystyle{plain}
\bibliography{ReferenceBudgetedColoring}
\end{document}

%% file: abstract_mod.tex
We introduce a variant of the graph coloring problem, which we denote as {\sc Budgeted Coloring Problem} (\bcp). Given a graph $G$, an integer $c$ and an ordered list of integers $\{b_1, b_2, \ldots, b_c\}$, \bcp asks whether there exists a proper coloring of $G$ where the $i$-th color is used to color at most $b_i$ many vertices. This problem generalizes two well-studied graph coloring problems, {\sc Bounded Coloring Problem} (\bocp) and {\sc Equitable Coloring Problem} (\ecp) and as in the case of other coloring problems, it is \nph even for constant values of $c$.
So we study \bcp  under the paradigm of parameterized complexity, particularly with respect to (structural) parameters that specify how far (the deletion distance) the input graph is from a tractable graph class.
\begin{itemize}
	\item
	We show that \bcp is \fpt (fixed-parameter tractable) parameterized by the vertex cover size. This generalizes a similar result for \ecp and immediately extends to the \bocp, which was earlier not known. 
	\item
	We show that \bcp is polynomial time solvable for cluster graphs generalizing a similar result for \ecp. However, we show that \bcp is \fpt, but unlikely to have polynomial kernel, when parameterized by the deletion distance to clique, contrasting the linear kernel for \ecp for the same parameter.
	\item
	While the \bocp is known to be polynomial time solvable on split graphs, we show that \bcp is \nph on split graphs. As \bocp is hard on bipartite graphs when $c>3$, the result follows for \bcp as well. We provide a dichotomy result by showing that \bcp is polynomial time solvable on bipartite graphs when $c=2$. We also show that \bcp is \nph on co-cluster graphs, contrasting the polynomial time algorithm for \ecp and \bocp.  
\end{itemize}
Finally we present an $\mathcal{O}^*(2^{|V(G)|})$ algorithm for the \bcp, generalizing the known algorithm with a similar bound for the standard chromatic number.

%% file: intro.tex
\section{Introduction}\label{sec:intro}
A proper vertex coloring of a graph $G$ is an assignment $\varphi:V(G)\xrightarrow[]{}[c]$, of colors to its vertices such that, for any edge $(u,v)\in E(G), ~\varphi(u)\neq \varphi(v)$, here $[c]=\{1,2, 3, \cdots, c\}$. Minimum number of colors used by any proper coloring of $G$ is the {\it chromatic number} of $G$ and denoted by $\chi(G)$. 
For any proper coloring $\varphi$, we denote the set of vertices which gets color $i$ by $V_i$, formally $V_i=\{v\in V(G)|\varphi(v)=i\}$.

We introduce a generalization of a well-studied variant called {\sc Bounded Coloring Problem} (\bocp) which asks for a proper coloring of $G$ with each $|V_i|$ bounded by a given constant $d$. One motivation comes from
the problem {\sc Bin Packing Problem with Conflicts} (BPPC)~\cite{binpacking}. Here we are given a set of $n$ unit sized objects to be packed into at most $c$ bins of size at most $d$ each, except that some pairs of objects can not be placed in the same bin. This conflict information can be captured by a conflict graph and the problem is exactly an instance of \bocp. 
We consider a natural generalization where the bins have different, but given, sizes. The associated graph coloring problem which we call {\sc Budgeted Coloring Problem} (\bcp) is as follows.

\defproblem{{\sc Budgeted Coloring Problem} (\bcp)}{An undirected graph $G(V,E)$, an integer $c$ and an ordered list $\B=\{b_{1}, b_{2}, \ldots, b_{c}\}$}{Does there exist a proper coloring $\varphi: V(G)\rightarrow [c]$ of $G$ such that $|V_i|\leq b_i$ for all $1\leq i\leq c$?}   

Given $G$ and $\B$ we denote any proper coloring $\varphi$ of $G$ a { \sc  proper  budgeted coloring } if  $|V_i|\leq b_i$ for all $1\leq i\leq c$.
\bcp also generalizes another well-studied variant, the {\sc Equitable Coloring Problem} (\ecp). 
Given a graph $G$ and an integer $c$, \ecp asks whether $G$ can be colored with $c$ colors such that for each $i$, $|V_i|=\floor*{\frac{n}{c}}$ or $|V_i|=\ceil*{\frac{n}{c}}$.


As coloring is hard even when $c$ is $3$ for general graphs \cite{GareyJohnson}, we can rule out any fixed parameter tractable (\fpt) algorithm parameterized by the number of colors for \bcp for general graphs. In this paper, we study the complexity of \bcp for restricted graph classes and present \fpt algorithms parameterized by some structural parameters. 

Throughout the paper, we denote the number of colors by $c$. We follow the symbols and notations of graph theory and parameterized complexity theory as in the textbooks \cite{diestel2012graph} and \cite{ParamAlgorithms15b}, respectively. 


It follows from the definition of \bocp that for the graph classes for which \bcp has a polynomial time or fixed-parameter tractable algorithms, \bocp also has a similar algorithm. On the other hand, if \bocp is \nph (\wih for a parameter) on a graph class then \bcp is \nph (\wih for the same parameter) too.
The following extension is not obvious though not difficult.

\begin{lemma}$^\star$ \label{lemma:allsame}
	If there exists an algorithm to solve \bcp  for a graph class $\Ge$ with respect to some parameter $k$ in time $f(k,|V(G)|)$, then we can solve \ecp in time $f(k,|V(G)|)$ for any graph $G\in\Ge$.
	
\end{lemma}
The following two corollaries are immediate from Lemma \ref{lemma:allsame}.

\begin{corollary}
	If for some graph class $\Ge$ (and a parameter $k$), \bcp is polynomial time (\fpt for the parameter $k$) then \bocp and \ecp are also polynomial time (\fpt for the same parameter respectively) for $\Ge$. 
\end{corollary}

\begin{corollary}\label{cor:hard}
	If for some graph class $\Ge$ (and a parameter $k$), \bocp or \ecp is \nph (\wih for the parameter $k$), then  \bcp is \nph (\wih for the same parameter respectively) for $\Ge$. 
	
\end{corollary}

It follows from Corollary \ref{cor:hard} that \bcp is \nph on interval graphs and co-graphs as \ecp is \nph on these class of graphs\cite{gomes2019structural}.  Similarly as \ecp is \woh on bounded treewidth graphs~\cite{FELLOWStw}, \bcp is also \woh parameterized by treewidth.

\begin{table}[!htb]
	\begin{center}
		
		\begin{tabular}{|p{0.12\textwidth}|C{0.17\textwidth}|C{0.21\textwidth}|C{0.25\textwidth}|C{0.25\textwidth}|} 
			\hline
			\textbf{Graph Class} & \textbf{Chromatic Number} & \textbf{Equitable Coloring} & \textbf{Bounded Coloring} & \textbf{Budgeted Coloring} \\
			\hline
			Bipartite & {\sf Polynomial Time} & \nph even for three colors $^\star$ & \nph  even for three colors \cite{BodlaenderJ93} & \nph  even for three colors \\
			\hline
			Cluster & {\sf Polynomial Time} & {\sf Polynomial Time} \cite{gomes2019structural}  & {\sf Polynomial Time} & {\sf Polynomial Time}$^\star$ \\
			
			\hline
			Split & {\sf Polynomial Time} & OPEN & $\textsf{Polynomial Time}$ \cite{chen1995equitable} & \nph $^\star$ \\
						\hline
			Co-Cluster & {\sf Polynomial Time} & {\sf Polynomial Time} & {\sf Polynomial Time} & \nph $^\star$ \\

			\hline
			
		\end{tabular}
		\vspace{5pt}
		\caption{Summary of Results in Different Graph Classes; Results marked with $\star$ are in this paper.}\label{tab:table1}
	\end{center}
\end{table}

\begin{table}[!htb]
	\begin{center}
		
		\centering
		\begin{tabular}{|p{0.4\textwidth}|C{0.15\textwidth}|C{0.15\textwidth}|C{0.15\textwidth}|C{0.15\textwidth}|} 
			\hline
			\textbf{Parameters} & \textbf{Chromatic Number} & \textbf{Equitable Coloring} & \textbf{Bounded Coloring} & \textbf{Budgeted Coloring}  \\
			\hline
			Cluster Vertex Deletion (CVD) Size & \fpt & \fpt \cite{gomes2019structural}& OPEN & OPEN \\
			\hline
			CVD Size + Number of Colors & \fpt  & \fpt \cite{gomes2019structural}  & \fpt $^\star$ & \fpt $^\star$   \\
			\hline
			CVD Number + Number of Clusters &\fpt   &  \fpt \cite{gomes2019structural}  &  \fpt $^\star$   & \fpt $^\star$   \\
			\hline
			Vertex Cover & \fpt~\cite{bodlander_cross}&  \fpt\cite{fiala2011parameterized} & \fpt $^\star$  & \fpt $^\star$, Polynomial kernel unlikely $^\star$  \\
			\hline
			Distance to Clique & \fpt~\cite{paramclique}&  Linear Kernel \cite{gomes2019structural} & \fpt $^\star$  & \fpt $^\star$,  Polynomial Kernel Unlikely $^\star$ \\
			\hline
			
		\end{tabular}
		\vspace{5pt}
		\caption{ Summary of Results in Parameterized Setting; Results marked $\star$ are in this paper.}\label{tab:table2}
	\end{center}
\end{table}

\noindent {\bf Our Results:} In this paper, we first show \nph and polynomial time results for \bcp in some graph classes. For the most part, we design fixed-parameter tractable algorithms parameterized by cluster vertex deletion set size (i.e. minimum number of vertices whose removal makes the graph a cluster graph -- a collection of cliques) and vertex cover size. The results (including previously known results on \ecp and \bocp, for contrast) are summarized in Table~\ref{tab:table1} and Table~\ref{tab:table2}. We also give an $\mathcal{O}^*(2^{|V(G)|})$ exact algorithm for the problem generalizing the known algorithm with a similar bound for the standard chromatic number.

%% file: NotationsandPreliminaries.tex
\section{Notation and Preliminaries} \label{sec:not}

\noindent \textbf{Graph Notations}

All the graphs considered used in our study are simple (free form self loops and parallel edges), finite, and undirected. In a graph $G(V,E)$, for any vertex $v \in V$, we denote its open neighborhood as $N(v)=\{u \vert (uv) \in E\}$. Define $N_G(v_i)= N(v_i)\cap V(G)$. For any vertex $v \in V$, its degree is denoted by $deg(v)=|N(v)|$. For any subset $S \subset V$, and $v \in V \setminus S$, degree of $v$ in $S$ is denoted as $deg_{S}(v)=|N(v)\cap S|$. For any $S \subseteq V$, its open neighborhood is denoted as $N(S)$ and defined as $N(S)=\{u \in V \setminus S| \exists v \in S \text{ and } (uv) \in E \}$ and $N[S]=N(S)\cup S$. For any subset of vertices $S$ of $V$, the subgraph induced by $S$ is denoted as $G[S]$. For a graph $G$ and a subset of vertices $S \subseteq V$, by $G-S$ we denote the graph induced by the vertex set $G \setminus S$. Given a graph $G$, its complement graph $G^c$, is defined as the graph with same set of vertices but the edges in $G$ become non-edges in $G^c$ and non-edges in $G$ become edges in $G^c$.  We denote $\K_n$ as complete graph with $n$ vertices and $\K_{n,n}$ as a complete bipartite graph with $n$ vertices in each part.  A subset of the vertices of a graph is said to be a vertex cover if for every edge at least one end vertex is in the subset.

\noindent \textbf{Graph Classes}
A graph is said to be {{\textsc{bipartite}}} if its vertex set can be partitioned into two subsets such that two endpoints of every edge belong to two different partitions. A {\sc star} is a complete bipartite graph where one part contains one vertex only and the other part contains all other vertices. A {\sc broom} is a star where only one edge can be extended as a path.   A graph is said to be {\sc split} if its vertex set can be partitioned into a clique and an independent set. A graph is said to be a {\sc cluster} if it is a collection of disjoint cliques.  A {\sc chordal} graph is graph with no induced cycle of size at least $4$. A {\sc co-cluster} is complement graph of {\sc cluster} graph i.e., complete multipartite graph.

\noindent \textbf{Parameterized Complexity }
A { parameterization} of a problem is assigning an integer $\ell$ to each input instance, and we say that a parameterized problem is {\em fixed-parameter tractable 	(\fpt)} if there is an algorithm that solves the problem in time
$\mathcal{O}(f(\ell)\cdot |I|^{O(1)})$ (also written as $\mathcal{O}^*(f(\ell)))$, where $|I|$ is the size of the input and $f$ is an
arbitrary computable function depending on the parameter $\ell$
only. 
Another major research field in parameterized complexity is kernelization. 
A parameterized  problem is said to admit a {\em polynomial kernel} 
if any instance $(I,k)$ can be reduced  to an equivalent instance $(I',k')$, in polynomial time, with $\vert I' \vert$ and 
$k'$ bounded by a polynomial in $k$. 
The {\sf W} hierarchy is a collection of computational complexity classes. If a problem belongs to \woh it implies the non-existence of a \fpt-approximation algorithm for the problem under the standard parameterized complexity hypothesis \woh $\neq$ \fpt.
For more background, the reader is referred to the monographs \cite{ParamAlgorithms15b,DowneyFbook13,fomin2019kernelization,Niedermeier06}.

\noindent \textbf{Graph Parameters} 
The \textsc{cluster vertex deletion number} (CVD) of a graph is defined as the minimum number of vertices that need to be deleted such that the remaining graph becomes a cluster graph. Such a set of vertices is called a \textsc{cluster vertex deletion set}.   Boral et al. \cite{boral2016fast} have showed that given a graph $G$, there exist an \fpt algorithm to find out whether there exists a cluster vertex deletion set for $G$ of size at most $k$ or not with runtime $\mathcal{O}^{*}(1.9102^{k})$.  
A graph $G$ is said to be $k$  \textsc{distance to clique} if there exists a $A\subseteq V(G)$ with $|A|=k$ such that $G[V\setminus A]$ is a clique. 
As shown in \cite{gutin2019parameterized}, for any graph $G$, determining whetheror not  $G$ is $k$  \textsc{distance to clique}, can be computed in $\mathcal{O}^{*}(1.2738^{k})$ time and can be approximated within a factor of two.
Determining whether there is a vertex cover of size $k$, parameterized by $k$ is $\sf FPT$ and can be computed in $\mathcal{O}^{*}(1.2738^{k})$ time.

\noindent \textbf{Known \npc Problems}

\defproblem{\sc Subset Sum Problem}{A set of integers $X=\{x_{1}, x_{2}, \ldots, x_{p}\}$, and an integer $w$.}{Does there exist a subset $\mathcal{X} \subseteq X$ such that $\underset{x_{i} \in \mathcal{X}}{\sum} x_{i}=w$?}
\defproblem{\sc Dominating Set Problem}
{An undirected graph $G(V,E)$, and a positive integer $k$.}
{Is there a subset $S \subseteq V$ such that $|S|=k$ and $N[S]=V(G)$?}

\defproblem{\sc Bi-clique Problem}{A bipartite graph $G(V_{1}, V_{2}, E)$, and a positive integer $k$.}{ Is there a complete bipartite graph of size $k \times k$ as a subgraph of $G$?}

\defproblem{\sc Clique Problem}{An undirected graph $G(V,E)$, and a positive integer $k$.}{Is there a complete subgraph of $k$ vertices in $G$?}

\defproblem{\partp}{A set of integers $X=\{x_{1}, x_{2}, \ldots, x_{p}\}$, and an integer $w$.}{Does there exist a partition of $X$ into $\frac{p}{3}$ triplets such that the elements in each triplets add up to $w$?}

\noindent \textbf{Known Results for Coloring}
It is a well-known fact that a bipartite graph is $2$-colorable and it can be checked in polynomial time.  For the class of cluster graphs, the chromatic number is the size of the largest clique, this can also be verified in polynomial time. Note that split graphs are a sub-class of chordal graphs. The classical coloring problem can be solved in polynomial time. Therefore, the coloring problem is also polynomial time solvable on split graphs. Next, we move to the domain of parameterized complexity.  Even though the classical coloring problem parameterized by CVD number is folklore {\sf FPT} result, for the completeness of our paper we briefly describe an approach to solve the problem in {\sf FPT} time. Due to ~\cite{gomes2019structural}, it is known that \ecp is {\sf FPT} parameterized by CVD. Next, by showing a  parameter preserving reduction to \ecp we prove that the coloring problem is \fpt when parameterized by CVD. 
\begin{lemma}
	Coloring problem is \fpt parameterized by CVD.
\end{lemma}
\begin{proof}
	Let $\I= (G,c)$ be a given instance of coloring problem. We create an instance $\I'=(G', c)$ of \ecp as follows. $V(G')= V(G)\cup W$, where $|W|=nc-n$, $E(G')= E(G)$. Observe that each vertex in $W$ is a singleton cluster. Thus the parameter remains same.
	
	Now, we prove that  $\I$ is a YES instance if and only if $\I'$ is also a YES instance. Assume $G$ admits a proper coloring with $c$ colors. Let each color $i$ is used to color $n_i$ many vertices in $G$. Note that, $|V(G')|=nc$. Thus each color must be used exactly $n$ times in an equitable coloring.  As $G'[W]$ is an independent set,  we can use any color to color the vertices in $W$. For each color $i$, we color $n-n_i$ vertices of $W$. Thus we have an equitable coloring of $G'$. Next assume that we have an equitable coloring of $G'$. As $G$ is a subgraph of $G'$, the graph $G$ also admits a proper coloring using at most $c$ colors. This completes the proof.    
\end{proof}

%% file: special.tex
\section{On Special Graphs}\label{sec:special}
In this section we derive the complexity of \bcp on cluster graphs (a collection of cliques) and co-cluster graphs (complete multipartite graphs), split graphs and bipartite graphs. The result on cluster graphs will be used in the next section when we generize to look at parameterization by cluster deletion set.
\subsection{On Cluster Graphs}\label{sec:cluster}
\bcp is trivial on cliques. Though \bcp can be solved in polynomial time on cluster graphs by constructing flow network~\cite{gomes2019structural}, we give an alternate simpler algorithm. Next, we have the following lemma.

\begin{lemma}\label{lemma:cluster_poly}
	Let $\I=(G,c,\B)$ be an instance of \bcp where $G$ is a cluster graph with a set $\{\K_1, \K_2,$ $ \cdots, \K_\ell\}$ of clusters sorted in non-increasing order of their sizes. 
	If $\I$ is a YES instance, then there exists a proper coloring of $G$ 
	where the largest cluster $\K_1$ is colored with $|V(\K_1)|$ colors having the largest $|V(\K_1)|$ budgets.
\end{lemma}

\begin{proof}
	Without loss of generality, assume that the budgets in $\B$ are sorted in non-increasing order i.e. $b_1\geq b_2 \geq \cdots \geq b_c$. For any proper budgeted coloring $\varphi$, with slight abuse of notation let us define $\varphi(\K_1)$ as the set of colors used to color the vertices of the clique $\K_1$, formally $\varphi(\K_1)=\{ i\in [c]|~\exists v\in V(\K_1) \text{ such that } \varphi(v)=i\}$.    
	For any proper budgeted coloring $\varphi$, $\iota(\varphi)$ denotes the largest integer such that all the colors in $\{1,2, \cdots ,\iota(\varphi)\}$ are used to color some vertex of $\K_1$, i.e. $\iota(\varphi)=\max_j \{j|~[j]\subseteq \varphi(\K_1)\}$. Let $\varphi^*$ be the proper budgeted coloring which maximizes $\iota(\varphi)$ over all possible proper budgeted coloring $\varphi$. Observe that if $\iota(\varphi^*)=|V(\K_1)|$, then nothing to prove. Thus let us assume that $\iota(\varphi^*)=i< |V(\K_1)|$ and there exist a color $j>i$ which is used to color one vertex say $v$ of $\K_1$ in $\varphi^*$. Observe the color $i+1$ is not used to color any vertex of $\K_1$. If $|V_{i+1}|<b_{i+1}$ then we can color $v$ with $i+1$ and produce a proper budgeted coloring with greater $\iota$ value, contradicting the maximality of $\varphi^*$. Thus we assume that $|V_{i+1}|=b_{i+1}$. Observe that $j>i+1$ hence $b_{i+1}\geq b_j$. As $j\in \varphi^*(\K_1)$ and $(i+1)\notin \varphi^*(\K_1)$  therefore there exist a cluster $\K_\ell$ such that $(i+1)\in \varphi^*(\K_\ell)$ and $j\notin \varphi^*(\K_\ell)$. Let $u\in \K_\ell$ be the vertex such that $\varphi^*(u)=j$. By exchanging colors of $u$ and $v$ we can  produce a proper budgeted coloring with greater $\iota$ value, contradicting the maximality of $\varphi^*$. Hence contradiction and the claim holds.
\end{proof}
\begin{algorithm}
	\DontPrintSemicolon 
	\KwIn{$(G,c,\B)$ where $G$ is cluster graph with set of clusters $\K= \{\K_1, \K_2,$ $ \cdots, \K_\ell\}$, $c$ colors $\{1,2, 3, \cdots, c\}$ and
		$\B=\{b_1,b_2, b_3, \cdots, b_c\}$ be the budgets of the colors. }
	\KwOut{Returns a proper coloring of $G$ respecting the budgets if the given instance is YES instance otherwise, returns NO.}
	$K \gets $ sorted list of cliques based on their sizes in non-increasing order.\\
	$B \gets$ sorted list of budgets in non-increasing order.\\
	\For{$i \gets 1$ \textbf{to} $\ell$}{
		$k_i \gets $ the size of $K[i]$\\
		\For{$j \gets 1$ \textbf{to} $k_i$}{
			$v_{ij} \gets $ $j$-th vertex of $i$-th clique\\
			$c \gets$ color corresponding to budget $B[j]$\\
			\If {$B[j] \geq 1$} {
				$\mathbb{C}(v_{ij}) \gets c$\;
				$B[j] \gets B[j]-1$ 
			}
			\Else {\Return {NO}}
		}
		sort $B$ in non-increasing order.
	}
	\Return{$\mathbb{C}(V(G))$}\;
	\caption{\bcp on cluster graphs}
	\label{algo:cluster}
\end{algorithm}

Lemma \label{lemma:cluster_poly} leads to the greedy Algorithm~\ref{algo:cluster}, establishing the following theorem.
\begin{theorem}
	\bcp on the class of Cluster Graphs can be solved in polynomial time.
\end{theorem}

\subsection{On Co-Cluster Graphs}\label{sec:co-cluster}
In this section we present another contrasting result. While \ecp is known to be polynomial time solvable on co-cluster graphs, we prove that \bcp is \nph on this class of graphs. It is not hard to verify that \bocp can also be solved in polynomial time. Now, we prove the following theorem by showing a reduction from \partp. 

\begin{theorem}\label{th:co-cluster_hard}
\bcp on co-cluster graphs is \nph.
\end{theorem}

\begin{proof}
 Let $\I=(S, w)$ be a \partp problem instance, where $S=\{x_1,x_2,x_2,\cdots, x_{3c}\}$ and $w=\frac{n}{c},~n=\sum x_i$. We create an instance $\I'=(G, 3c, \B)$ of \bcp as follows. Here $G$ is a co-cluster graph with $c$ parts, where each part has $3n+w$ many vertices and $\B=\{n+x_1, n+x_2, n+x_3, \cdots, n+x_{3c} \}$. Next, we show that $\I$ is a YES instance if and only if $\I'$ is a YES instance. 

    Let $\I$ is a YES instance, then there are $c$ sets containing three elements each, adding up to $w$.  We can color each part of G by using the corresponding three colors whose budgets add up to $3n+w$.  Thus $\I'$ is also a YES instance.

    Now, assume that $\I'$ is a YES instance. Observe that, one color can be used in only one part. If each part gets exactly three colors then we have nothing to prove. Therefore we consider the case when there exists at least one part that gets at most two colors or at least four colors. For the first case, say $i$ and $j$ be that two colors. As $2n+x_i+x_j< 3n+w$, all the vertices in that part is not colored with color $i$ and $j$. For the second case, observe that as $4n>3n+w$, some color is not fully used to color that part.  Also observe that the budgets add up to the number of the vertices in $G$. As a result, there is at least one part that contains some uncolored vertices, which is contradiction.    

Note that, \partp is \nph even when all the elements are polynomial in the number of elements in the set ~\cite{GareyJohnson}. Thus, we can safely assume that each $x_i$ is polynomial in $3c$. Hence, the reduction is also polynomial time. 
\end{proof}

\subsection{On Split Graphs}\label{sec:split}
It is known that \bocp is polynomial-time solvable on the class of Split Graphs~\cite{chen1994equitable}. Surprisingly it turns out that \bcp is \nph on Split Graphs.  We give a reduction from the \dsp and prove the following theorem. 
\begin{theorem} \label{SplitHard}
	The \bcp is \nph on Split Graphs.
\end{theorem}
\begin{proof}
	Let, $\I=(G,k)$ be an arbitrary instance of the \dsp, where $V(G)=\{w_1, w_2, \cdots , w_n\}$ and $k \leq n$. We construct a \bcp instance $\I'=(G', n, \mathcal{B})$, where $G'$ is a split graph as follows.
	$V(G')= C\cup I$, where $C=\{u_1, u_2, \ldots, u_n\}$ be the clique and $I=\{v_1, v_2, \ldots, v_n\}$ be the independent set. $E(G')=E_1\cup E_2$, where $E_1=\{(u_i,u_j)| ~\forall i\neq j\}$ and $E_2=\{(u_i,v_j)| (w_i,w_j)\notin E(G) \}$. The number of colors $c$ is $n$, with budgets $b_i=n+1$ for $i\in [k]$ and the budget for the remaining colors is $1$.
	
	Next, we show that the graph $G$ has a dominating set of size $k$ if and only if there exists a proper budgeted coloring of $G'$.
	Assume that the graph $G$ has a dominating set $D$ of size $k$. Without loss of generality assume $D=\{w_1,w_2, \cdots , w_k\}$. Color $u_i$ with color $i$ where $i \in [k]$. As $D$ is a dominating set in $G$, each vertex in $I$ is adjacent to some vertex $d$ of $D$ in $G$, and hence is non-adjacent to that vertex in $G'$, and so can be colored with the color of $d$. Thus we can color every vertex in $I$ with the first $k$ colors. Now, we have only $n-k$ uncolored vertices in $C$. Color them with the last $n-k$ colors. Hence we have a proper budgeted coloring of $G'$.

	To prove the converse, note that in any proper budgeted coloring of $G'$, all $n$ colors must be used to color the vertices in $C$. 
	Without loss of generality assume that $u_i$ gets the color $i$, for $i=1$ to $n$. 
	Therefore the vertices in $I$ are colored with the first $k$ colors (as only they have budgets more than $1$). 
	Thus $\cup_{i\in [k]}\overline{N}(u_i)= I$. Hence $\{w_1, w_2, \ldots, w_k\}$ forms a dominating set in $G$. This completes the proof.		
\end{proof}

It is known that the dominating set problem is \wth with respect to the solution size ($k$) as a parameter. In the reduction described in Theorem \ref{SplitHard}, the number of colors having a budget greater than one is also $k$. Hence, this reduction is parameter preserving for the parameter `number of colors with a budget greater than one' leading to the following corollary.
\begin{corollary}
	\bcp parameterized by the number of colors with a budget greater than one is \wth.
\end{corollary}

\subsection{On Bipartite Graphs} \label{sec:bipartite}

It was known that the bounded coloring problem is \nph on bipartite graphs when $c \geq 3$~\cite{BodlaenderJ93}. We extend this result for \ecp thus proving hardness for \ecp as well as \bcp (the later result follows from Corollary $2$.)

\begin{theorem}\label{th:equitable_bipartite}
	\ecp on Bipartite Graphs is \nph even when $c=3$. 
\end{theorem}
\begin{proof}
	We show a reduction from the biclique problem. Given a biclique problem instance $\I=(G(V_1 \cup V_2, E),k)$ where, $|V_1|=|V_2|=n$, the question is whether $G$ has $\K_{k,k}$ as a subgraph? We construct a \ecp instance $\I'=(G'(V'_1\cup V'_2, E'),3)$ as follows. $V'_1=V_1\cup \{a\}$, $V'_2=V_1 \cup W \cup \{b\}$ where $a$, $b$ are two new vertices and $|W|=n-3k+1$. The edge set $E'=\{E_1\cup E_2\cup E_3\}$ where, $E_1=\{(u,v) | u\in V_1, v\in V_2\ \text{and} \ (u,v)\notin E \}$, $E_2=\{(a,v)|v \in V'_2\}$ and $E_3=\{(b,u)|u \in V'_1\}$. Also we have three colors $1,2$ and $3$. Since $|V(G')|=3n-3k+3$, to achieve proper equitable coloring we must use each color exactly $n-k+1$ times. 
		\begin{claim}
			$\I$ is a YES instance iff $\mathcal{I'}$ is a YES instance.
		\end{claim}
		\begin{proof}
	Let $\I$ be a YES instance, i.e. $G$ has a $\K_{k,k}$ as a subgraph. Without loss of generality let $\{u_1, u_2, u_3, \cdots, u_k\}\in V_1$ and $\{v_1, v_2, v_3, \cdots, v_k\}\in V_2$ form the $\K_{k,k}$ in $G$. Thus in $G'$, $(u_i, v_j)\notin E',~\forall i,j \in [k]$. Use color $1$ to color all these vertices and the vertices in $W$. Color remaining $n-k+1$ vertices in both $V_1'$ and $V_2'$ with color $2$ and color $3$, respectively. Observe that we have a proper equitable coloring of $G'$.
	
	Next, assume that $G'$ admits a proper equitable coloring. Without loss of generality,  assume color $1$ and $2$ are used to color the vertices $a$ and $b$. Since $a$ is adjacent to all the vertices in $V_2'$, color $1$ must be entirely used at vertices in $V_1'$. Also, color $2$ can color at most $n-k+1$ many vertices of $V_2'$, more precisely, at most $n-k$ many vertices of $V_2$. Thus at least $k$ vertices from both the parts $V_1$ and $V_2$ must be colored with color $3$. Hence they are forming a $\K_{k,k}$ in $G$. This completes the proof.
	\end{proof}
\end{proof}

\subsubsection{When $c=2$.}
While the \bocp has been known to be hard for bipartite graphs when $c \geq 3$~\cite{BodlaenderJ93}, its complexity for $c=2$ (especially in disconnected bipartite graphs) doesn't seem obvious. We show that even the more general \bcp is polynomial time solvable when $c=2$
by a suitable reduction to the {\it subset sum problem} to prove the following.
\begin{theorem}\label{bipartitec=2}
The \bcp on the class of bipartite graphs can be solved in polynomial time when $c=2$.
\end{theorem}
\begin{proof}
	Let, $G(V_1\cup V_2, E)$ be the given bipartite graph where $|V_1\cup V_2|=n$, $b_1$ and $b_2$ be the budgets associated with the two colors. 
	Observe that, \bcp is polynomial time solvable when the $G$ is connected as we just have to ensure that the two parts of the partition satisfy the budget constraints. 
	When $G$ is disconnected, let $P_1, P_2, P_3, \cdots, P_\ell, \ell \geq 2$ be the connected components in $G$ where each connected component in the graph $P_i=(V^i_1 \cup V^i_2, E_i)$. Without loss of generality assume that $|V^i_1| \geq |V^i_2|$ for all $i$. In any proper coloring, at least budget $\sum_{i\in [\ell]}V^i_2=x$ is required for both the colors. Thus we can subtract $x$ from the budgets $b_1$ and $b_2$ (and if any of the budgets is less than $x$, then the problem is a NO instance). Let, $b'_1=b_1-x$ and $b'_2=b_2-x$ be the modified budgets for the colors $1$ and $2$ respectively. 
	Let $y_i= |V^i_1|- |V^i_2|$ for all $i \in [\ell]$. Now, each component $P_i$ has $y_i$ many uncolored vertices that can be colored with any of the two colors. 
	
	Let $S=\{y_1, y_2, y_3, \cdots, y_\ell\}$ and we would like to find $X \subseteq S$ such that the sum of the elements of $X$ is at most $b'_1$ and the sum of the elements of $S\setminus X$ is at most $b'_2$. Observe that, this is exactly the subset sum problem where each $y_i$ is upper bounded by $n$. 
		It is known (see, for example ~\cite{PISINGER19991}) that subset sum problem can be solved in time $\mathcal{O}(nw)$ where $n$ is the number of elements in the given set when each element is bounded by $w$. Applying that algorithm for our problem proves the theorem.
\end{proof}

\subsection{On Paths Graphs}
As the \bcp is \nph on bipartite graphs when $c>2$, we study the problem on a sub-class. Here we study the problem on the class of paths and show that the problem is polynomial time solvable even when $c>2$ in the following lemma.

\begin{lemma}
The problem instance $(G(V,E), \mathcal{C}, \mathcal{A})$ is polynomial time solvable for the family $\mathcal{P}$ of paths.
\end{lemma}

\begin{proof}
We are given with a path $P$ of $n$ vertices. A necessary condition is that there are at least two colors. If any color has a budget of more than $\lceil \frac{n}{2} \rceil$, truncate to $\lceil \frac{n}{2} \rceil$
(as we can not use more colors than this). Now we claim that if the total budget (after truncation) is at least $n$, then the path is two colorable and it can not be otherwise. The algorithm in details as follows:

First we sort the colors in descending order based on their budgets and let the sorted list be $\mathcal{C'}$. The content of $\mathcal{C}'$ are as follows: $\mathcal{C}'=\{(i,b_i): \forall i \in [c]\}$ Now we have the following two cases:

\textbf{Case1:} The budget $b_i \leq \ceil*{\frac{n}{2}}, \forall i \in [c]$.

The given instance is a trivial NO instance if total budgets of available color is less than $n$. Now we check for the $\ceil*{\frac{n}{2}}$-th entry in $\mathcal{C'}$. Then we color the odd vertices of the path $P$ with first $\ceil*{\frac{n}{2}}$-th color in $\mathcal{C'}$ and even vertices with $\ceil*{\frac{n}{2}}+1$ to $n$-th color in $\mathcal{C'}$. As $b_i \leq \ceil*{\frac{n}{2}}, \forall i \in [c]$, this will return a proper coloring of $P$.\\

\textbf{Case2:} There exist an $i \in [c]$ such that, the budget $b_i > \ceil*{\frac{n}{2}}$.

Again we can say that the given instance is a trivial NO instance if total budget of available color is less than $n$. If not then start decreasing the availability of such color let $i$, until $b_i \leq \ceil*{\frac{n}{2}}$. Do that for all such color. Now the reduced instance is same as the instance given in Case 1.
\end{proof}

\subsection{On Brooms}
In this section we prove that the \bcp is polynomial time solvable for the class of brooms. Let $G(V,E)$ be a broom such that $V=U\cup W$, $U=\{u_1,u_2, \cdots, u_p\}$, $W=\{w_1,w_2, \cdots, w_q\}$ and $E=\{(w_i,w_{i+1})|~i\in[q-1]\}\cup \{(u_j,w_1)|~j\in[p]\}$. In other words, $W$ induces a path of length $q-1$, $U$ induces an independent set of cardinality $p$ and every vertex in $U$ is adjacent to only $w_1$ in the graph. Let $\I= (G,c,\B)$ be the given problem instance of the \bcp on brooms. Next, we prove that if the given instance is a YES instance then there exist a coloring that uses the minimum budgeted color at $w_1$.
\begin{lemma}\label{lemma:broom_mincolor}
If $\I$ be a YES instance then there is coloring $\psi$, such that in $\psi$, the vertex $w_1$  is colored with minimum budgeted color.
\end{lemma}
\begin{proof}
Without loss of generality assume that $1$ be the minimum budgeted color in the given instance. Let such that $a$ is used at $w_1$ in a proper coloring $\zeta$. Here we give another proper coloring $\psi$, that uses color $1$ in then vertex $w_1$. We denote $V_i$ as the set of vertices colored with color $i$ and $\zeta(v)$ as the color of $v$ in the coloring $\zeta$. Observe that for any color $a$, if $|V_a|\leq |V_1|$ then, by switching the color in $V_a$ and $V_1$ we can get a new coloring $\psi$, which is a proper coloring.  Now, we prove the lemma in the following two cases.

\begin{description}
\item[Case 1: $|N(w_1)\cap V_1|=1$.] Without loss of generality assume $N(w_1)\cap V_1=\{u_1\}$ in the coloring $\zeta$. In the new coloring $\psi$ switch the color used at $w_1$ and $u_1$, i.e., $\psi(w_1)=1$ and $\psi(u_1)=a$. Observe that $\psi$ is also a proper coloring of $G$ as only one neighbor of $w_1$ is colored with $1$.

\item[Case 2: $|N(w_1)\cap V_1|\geq 2$.] Without loss of generality assume $|N(w_1)\cap V_1|=k,~2\leq k \leq p+1$ in the coloring $\zeta$. As $|V_a|>|V_1|$, there is at least $k$ many vertices in $W\cap V_a$ such that no neighbors of them is colored with color $1$. Let the set be  $W'$. Now, in the new coloring $\psi$, we color $N(w_1)\cap V_1$ with color $a$, $w_1$ and any $k-1$ many vertices from $W'$ with color $1$. Observe that, in both the cases $\psi$ is a proper coloring of $G$ respecting the budgets.

\item[Case 3: $|N(w_1)\cap V_1|=0$.] If the color $1$ is not used in the coloring $\zeta$ of $G$, then color $w_1$ with $1$. Otherwise, pick arbitrarily one vertex from $U$, say $u_1$. Observe that $\zeta(u_1)$ and $\zeta(w_1)$ both can not be the same as $\zeta(w_q)$. If $\zeta(u_1)\neq \zeta(w_q)$ then, perform circular shifting of color starting from $u_1$ to $w_q$. In the new coloring $\psi(u_1)=\zeta(w_1), \psi(w_1)=\zeta(w_2),\cdots, \psi(w_{q-1})=\zeta(w_q), \psi(w_q)=\zeta(u_1)$. Continue the process till $\psi(w_1)=1$. If $\zeta(u_1)= \zeta(w_q) $ and $\zeta(w_1)\neq \zeta(w_q)$ then perform the similar circular shifting starting from $w_1$ to $w_q$ till $\psi(w_1)=1$. Observe that $\psi$ gives us proper coloring of $G$. This completes the proof.
\end{description}

\end{proof}
Now, we describe a coloring scheme using the fact that the minimum budgeted color is used to color the vertex $w_1$. Note that, we use a color at most $\lceil \frac{q}{2} \rceil$ many times to color the vertices in $W$. Color $\min(b_1, \lceil \frac{q}{2} \rceil)$ many odd vertices of the path using the color $1$. Next, pick an arbitrary color $i$ and color the remaining odd vertices in the path until we use that color $\min(b_i, \lfloor \frac{q}{2} \rfloor)$ many times.  If the number of remaining odd vertices in the path is less than $\min(b_i, \lfloor \frac{q}{2} \rfloor)$ then, use the rest budget to color the even vertices in the path. Repeat this process until all the vertices in the path are colored. Observe that any vertex in $U$ is non-adjacent to vertices in $W\setminus \{w_1\}$. Thus any color other than $1$ used in $W$ can be used to color the vertices in $U$ (if we have not used the budget completely). Use the colors (except $1$) with the remaining budgets to color the vertices in $U$. If we have enough budgets return YES, otherwise return NO. Moreover, observe that the entire process can be done in polynomial time. Thus we have the following theorem.

\begin{theorem}
The \bcp is polynomial time solvable on the class of Brooms.
\end{theorem}

%% file: fpt.tex

\section{Structural Parameterization of \bcp} \label{Sec:FPT}

In this section, we address the parameterized complexity of \bcp with respect to some structural parameters, parameters that measure the (deletion) distance to a tractable graph class.  

\subsection{\bcp parameterized by Cluster Vertex Deletion Number (CVD)}\label{sec:param_CVD}
Recall from Section~\ref{sec:cluster} that \bcp is polynomial time solvable on cluster graphs. So, here we ask what if the input graph is $k$-vertices away from a cluster graph i.e., deleting $k$ vertices from the graph makes the graphs a cluster graph? We show that \bcp is \fpt parameterized by $k$ and an additional parameter which is the number of colors or the number of clusters in the resulting graph. 
We remark that we do not need to assume that we are given the deletion set $S$, as we can find it in $\mathcal{O}^*(1.9102^k)$ time~\cite{boral2016fast}. 

Consider any instance of $\I=(G,c,\B)$ of \bcp. Let $S$ be the cluster vertex deletion set of $G$ and $\Pe=\{P_1,P_2, \cdots, P_{\ell}\}$ be a partition of $S$. Let $\alpha$ be any proper coloring of $S$ such that for each part $P_i$ and any two vertices $u,v \in P_i$, $\alpha(u)=\alpha(v)$. Now, we define a new instance $\I'= (G,c, \B, \Pe, \alpha)$ of   {\sc Extended Budgeted Coloring Problem}  (\ebcp) as follows. Given $\I'$, the \ebcp asks whether there exists a proper budgeted coloring $\beta$ of $G$ such that $\forall v \in S, ~\beta(v)=\alpha(v)$. Without loss of generality assume that, in $\alpha$ for each $i,~P_i$ is colored with color $i$. Though the following lemma follows form Lemma 1 and Lemma 2 in ~\cite{gomes2019structural}, we give a slightly different proof just for completeness.

\begin{lemma}\label{lemma:cvd+colorscorrectness}
	Given a partition $\Pe$ and its coloring $\alpha$ of cluster vertex deletion set $S$, \ebcp instance $\I'=(G,c, \B, \Pe, \alpha)$ can be solved in polynomial time. 
\end{lemma}
\begin{proof}
	Let $\K=\K_1 \cup \K_2 \cup \K_3 \cup \cdots \cup \K_m$ be the clusters in $G-S$. Now, we construct a flow network as follows.
	For each part $P_i$ create one vertex $v_i$ and for each unused color $a$ in the coloring $\alpha$ create a vertex $v_{a+\ell}$, where $\ell$ is the number of parts in $\Pe$. 
	For each $(v_i,\K_j)$ pair, add a vertex $u_{ij}$ (see Figure~\ref{fig:CVD+Col}). Join each $v_i$ and $u_{ij}$ for all $j \in [m]$. 
	Connect each $u_{ij},~i\in [\ell]$ and $w\in \K_j$ by an edge, if $P_i\cup \{w\}$ is independent in $G$. Also connect each $u_{ij}$ and $w\in \K_j$ by an edge, if $~i> [\ell]$. Now, we add two more vertices $s$ and $t$ such that, every vertex $v_i$ is adjacent to $s$ and every vertex in $\K$ is adjacent to $t$. Next, we set the capacities of the edges. Capacity of every $(s, v_i)$ edge is $b'_i= b_i - |P_i|$ if $i \in [\ell]$ and is $b'_i= b_i$, otherwise. All other edges have capacity one. Next, we prove the following claim.
	
	\begin{figure}[!htb]
		\centering
		\includegraphics[width= 6 cm]{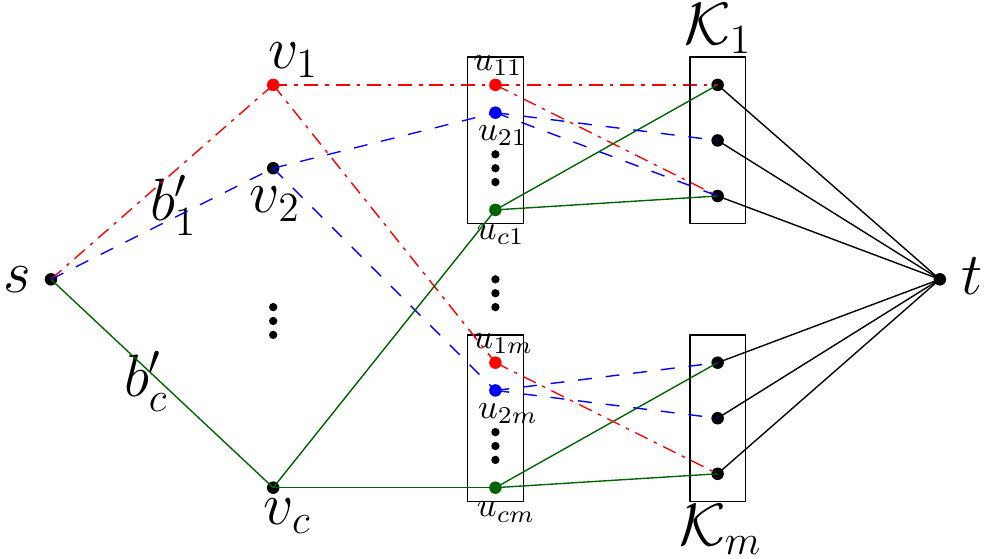}
		\caption{Construction of the flow instance for \bcp parameterized by CVD and the number of colors }
		\label{fig:CVD+Col}
	\end{figure}
	
	\begin{claim}
		\emph{ There exists a solution (proper budgeted coloring $\beta$) for the \ebcp instance $\I'$ if and only if the constructed flow network has an $(s,t)$-flow of value $n-k$, where $n$ is the number of vertices in $G$.}
	\end{claim}
	
	\begin{proof}
		Let us first assume that there exists a solution to the \ebcp instance $\I'$. Let $\beta$ be the coloring of $G$ such that, $\forall u \in S,~\beta(u)=\alpha(u)$.  Now we fix the flow in the flow network. Let us denote for the coloring $\beta$, for each $i, |V_i\cap \K|=x_i$. Let for an edge $(u,v), ~ f(u,v)$ denotes the flow of the edge. Set $f(s,v_i)=x_i$ for all $i$. If there exists a vertex $y\in V_i \cap \K_j$ then set $f(v_i,u_{ij})=1$, $f(u_{ij},y)=1$ and $0$ otherwise. As $\beta$ is a proper budgeted coloring of $G$, every vertex in $\K$ gets a color. Thus in this way for each vertex $w\in \K_j,~\forall j$ can make find one $i$ such that $f(u_{ij},w)=1$. Therefore $\sum_{w\in \K} f(w,t)=n-k$. As $\beta$ is a proper budgeted coloring, there does not exist any $\K_j,~j\in [m]$ that has two vertices $w,x$ such that $\beta(w)=\beta(x)$. Hence in the flow network also there are no $i\in [c]$ and $j\in [m]$, such that $f(u_{ij},w)=f(u_{ij},x)=1$, where $w,x \in \K_j$. Also note that  for each $i$, $x_i\leq b'_i$ and $f(s, v_i)=\sum_{j\in [m]}f(v_i,u_{ij})$. Therefore the capacity and the conservation constraints are satisfied. Hence we have a valid $(s,t)$-flow of value $n-k$.

		Now, we assume that the flow network has a $(s,t)$-flow of value $n-k$. That implies for all $w\in \K,~ f(w,t)=1$. For each vertex $w\in \K_j$ for all $j\in [m]$, find the $i$ such that $f(u_{ij},w)=1$. Color $w$ with color $i$. Now we have a coloring of the graph $G$.
		Next, we show that the coloring is a proper coloring of $G$. Let $w\in \K_j$ be a vertex that gets color $i$. We prove that, $w$ is neither adjacent to a vertex in $P_i$, nor another vertex in $\K_j$ gets the color $i$.  
		Observe that there is a path between $v_i$ and a vertex $x\in \K$ only if either $P_i \cup \{x\}$ is independent in $G$ or the color $i$ is not used at $S$. This satisfies the first condition. Furthermore, as we have a valid flow, due to the conservation constraint we know for a vertex $u_{ij}$, $\sum_{w\in \K_j} f(u_{ij},w)=1$. Thus no two vertices from the same clique get the same color in our coloring scheme, i.e., second condition is also satisfied. Again due to the capacity and conservation constraints, for each $i$ we have $f(s,v_i)=\sum_{j\in [m]}f(v_i,u_{ij})$. Therefore the budget constraints of \bcp are also satisfied. This completes the proof.
	\end{proof}   
	Moreover, observe that we can solve the flow instance in polynomial time. 
\end{proof}


Let $\Gamma$ be the number of colorings of $S$. Observe that $|\Gamma|=c^k$. If $(G,c,\B)$ is a YES instance of \bcp then there must exist one coloring $\alpha \in \Gamma$ such that we have a solution to the \ebcp with respect to $\alpha$. 
Thus from Lemma \ref{lemma:cvd+colorscorrectness} we have the following theorem.
\begin{theorem}\label{th:cvd+color}
	\bcp parameterized by cluster vertex deletion number and the number of colors can be solved in time $\mathcal{O}^*(c^k)$, where $c$ is the number of colors and $k$ is the cardinality of the cluster vertex deletion set.
\end{theorem}
Next, we prove that \bcp is \fpt parameterization by CVD and the number of clusters in $G \setminus S$. Observe that the number of clusters can be significantly smaller than the number of colors, for example when $G \setminus S$ is a large clique.
\subsubsection{Parameterization by CVD and the number of clusters}
\label{sec:cvd+clusters}


Let $S$ be the cluster vertex deletion set and $d$ be the number of clusters in $G - S$.
Let $\mathcal{P}=\{P_1, P_2, \cdots, P_{\ell}\}$ be any partition of $S$ into $\ell$ independent sets, and $D=\{d_1,d_2, \cdots, d_{\ell}\}$ be any ordered list of integers such that $0\leq d_i \leq d$ for each $i\in [\ell]$. Without loss of generality assume $|P_1|\geq |P_2| \geq \cdots \geq |P_{\ell}|$. For any $\Pe$, $D$  we define $\rgamma_{\Pe D}$ as follows. $\rgamma_{\Pe D}(P_1)$ be the least budgeted color with budget at least $|P_1|+d_1$. We inductively define $\rgamma_{\Pe D}(P_i)$ to be the the least budgeted color $a$ with budget at least $|P_i|+d_i$ such that for all $j<i$, $\rgamma_{\Pe D}(P_j)\neq a$.
Observe that each color is assigned to at most one part in a partition. We modify $\B$ and create a new list of budgets $\B_{\Pe D}=\{b'_1,b'_2,\cdots, b'_{\ell}\}$ as follows. $b'_a=|P_i|+d_i$ , if  there exists a $P_i$ such that $ \rgamma_{\Pe D}(P_i)= a $ and $b'_a=b_a$, otherwise.
Let $\mathcal{U}$ be the set of all pairs $(\Pe,D)$, where $\Pe$ be the all possible partition of $S$ into independent sets and  $D\in \{0,1,2, \cdots, d\}^{|\Pe|}$. 
If there is one $P_i$, such that $\rgamma_{\Pe D}(P_i)$ does not exist, discard the pair $(\Pe,D)$. Suppose we know the partition of $S$ induced by a feasible coloring; for each part in the partition we also know the number of vertices in $V\setminus S$ that get the same color. Given these two informations we can design a greedy algorithm to  solve the problem in polynomial time using Lemma~\ref{lemma:cvd+colorscorrectness}.  Towards that we prove the following lemma.

\begin{lemma}
	$\I=(G,c, \B)$ is a YES instance if and only if there exists at least one $(\Pe^*,D^*)\in \mathcal{U}$ such that $\I_{\Pe^* D^*}=(G,c,\B_{\Pe^* D^*}, \Pe^*, \rgamma_{\Pe^* D^*})$ is a YES instance. 
\end{lemma}

\begin{proof}
	Let us assume $\I$ is a YES instance of \bcp i.e., there is a proper budgeted coloring $\beta$ of $G$. Without loss of generality, assume that we have used the first $\ell$ colors to color the vertices of $S$. 
	Recall, $V_i=\{v|~\beta(v)=i\}$. We define $\Pe'=\{P'_1,P'_2,\cdots, P'_{\ell}\}$, where $P'_i=V_i\cap S$. Define $D'=\{d'_1,d'_2, d'_3, \cdots d'_{\ell}\}$ such that $d'_i=|V_i|-|P'_i|$ and $\B'=\{b'_1,b'_2, b'_3, \cdots, b'_c\}$, where  $b'_i= |V_i|$, if $i\in [\ell]$ and $b'_i=b_i$, otherwise.  Observe that, $(\Pe', D')\in \mathcal{U}$. As $\I$ is a YES instance, $\I'=(G,c,\B', \Pe', \beta)$ is also a YES instance of \ebcp.  Next we show that $\I_{\Pe' D'}=(G,c,\B', \Pe', \rgamma_{\Pe' D'})$ is a YES instance.
	
	Observe that if for all $1\leq j\leq \ell$, $\rgamma_{\Pe' D'}(P'_j) =\beta(P'_j)$  then nothing to prove. Let $P'_j$ be the largest part in $\Pe'$ such that $\rgamma_{\Pe' D'}(P_j) \neq\beta(P'_j)$. Let $\rgamma_{\Pe' D'}(P'_j)=x$ and $\beta(P'_j)=y$. By construction $b_x\leq b_y$. We can exchange colors $x$ and $y$ to construct a new coloring $\beta_1$ from $\beta$ which is also a feasible coloring for $\I'$. In $\beta_1$ one more part of $S$ receives same color as $\rgamma_{\Pe' D'}$. Thus applying the same step at most $\ell$ times we can create a feasible coloring $\beta_{\ell}$ such that for all $1\leq j\leq \ell$, $\rgamma_{\Pe' D'}(P'_j) =\beta_{\ell}(P'_j)$.  Thus proving $\I_{\Pe' D'}$ is a YES instance.

	While proving the converse, it is not hard to follow that if $\I_{\Pe, D,}$ is a YES instance then $\I$ is also a YES instance.
\end{proof}

Observe that, as there can be $k^k$ many partitions of $S$ and $d^k$ many choices of $D$. Thus $|\mathcal{U}|=\mathcal{O}(d^k \cdot 2^{k \log k})$. For each pair of $(\Pe,D)$ using Lemma~\ref{lemma:cvd+colorscorrectness}, we  can solve the problem in polynomial time. Therefore we have the following theorem.

\begin{theorem}\label{th:CVD+cluster}
	\bcp parameterized by $k$, the cluster vertex deletion number and $d$, the number of clusters can be solved in time $\mathcal{O}^*(d^k\cdot 2^{k\log k})$. 
\end{theorem}

\subsection{ \bcp Parameterized by the Distance to Clique}

First, observe that when the number of clusters is one, the problem reduces to the \bcp parameterized by the distance to a clique. Thus from Theorem~\ref{th:CVD+cluster} we get the following theorem as a corollary.
\begin{theorem}
	\bcp parameterized by the distance to clique can be solved in time $\mathcal{O}^{*}(2^{k\log k})$, where $k$  is the size of clique modulator.
\end{theorem}

Next, we prove that there is no polynomial kernel for \bcp parameterized by the distance to clique under standard complexity theoretic assumptions. We show a parameter preserving reduction from the clique problem (refer to Section~\ref{sec:not} for a formal definition) parameterized by vertex cover size.  
It is known that the clique problem parameterized by vertex cover does not admit a polynomial kernel unless \textsf{NP $\subseteq$ coNP/poly}~\cite{bodlander_cross}. 

\begin{theorem}\label{nopoly}
	\bcp parameterized by the distance to clique does not admit a polynomial kernel unless \textsf{NP $\subseteq$ coNP/poly}.
\end{theorem}
\begin{proof}
	
	Consider any instance of clique problem parameterized by vertex cover, $\I=(G,X,\ell)$. Here $G$ is a graph, $X\subseteq V(G)$ is a vertex cover and we would like to find out whether there exists a clique of size $\ell$ in $G$. Parameter is $k=|X|$. We construct the following instance $\I'=(G^c, n-\ell+1, \B)$ of \bcp as follows.  Here $G^c$ is the complement of $G$ and we set the budgets as follows. First $n-\ell$ colors have budget one and the last color has budget $\ell$. 
	
	Next, we prove that $G$ has a clique of size $\ell$ if and only if $\I'$  admits a proper budgeted coloring.
	Assume that $G$ has a clique of size $\ell$. It is an independent set in $G^c$, and uses the color with budget $\ell$ to color the independent set. There are $n-\ell$ colors with budget one and $n-\ell$ vertices left to color. Color them with all different colors. To prove the converse, assume $\I'$ is a YES instance. Observe that sum of the budgets of the colors is exactly $n$. So, the color with budget $\ell$ has been entirely used, which forms an independent set of size $\ell$ in $G^c$ and hence a clique in $G$.
	
	Observe that as $X$ is a vertex cover in $G$, $V(G)\setminus X$ is an independent set in $G$. Thus $V(G)\setminus X$ is a clique in $G^c$  and $X$ is a clique modulator in $G^c$. This completes the proof.
\end{proof}

\subsection{Budgeted Graph Coloring Parameterized by the Vertex Cover Size}\label{sec:param_VC}
In this section, we study the \bcp on the class of graphs that are $k$ vertices away from an independent set, i.e. on graphs that have a $k$-sized vertex cover. Observe that if a graph has a $k$-sized vertex cover, then it also has a $k$-sized CVD because a vertex cover is also a CVD. Thus the FPT results of section~\ref{sec:param_CVD} follow for vertex cover size as well. 
However, in this section, we give a stronger result by showing that \bcp is \fpt parameterized just by vertex cover, independent of any other parameter.

Let, $(G,c,\B)$ be an instance of the \bcp where the graph $G$ has a vertex cover $S$ of size $k$. So, $I= V(G) \setminus S$ is an independent set. 
Note that we do not need to assume that we are given $S$, as we can find $S$ in $\mathcal{O}^*(1.2738^k)$ time~\cite{param_VC}.

Our algorithm differs from the algorithm of Section~\ref{sec:param_CVD} (Proof of Theorem~\ref{th:cvd+color}) in only the first step of coloring the vertices of $S$. Instead of trying all possible colorings to color $S$ in the first step, we apply a greedy method. 

Let $\mathcal{P}=\{P_1, P_2, P_3, \cdots, P_\ell\}$ be the partition of $S$,  where $\ell \leq k$, and each part $P_i$ is independent. For each part $P_i$, we define the set of feasible colors for $P_i$ by $F_i$ where $F_i=\{j\in[c]|b_j\geq |P_i|\}$. Let $L_i$ be the set of least budgeted $\ell$ colors in $F_i$; if $|F_i|\leq \ell$, we set $L_i=F_i$.

Next we show that if there exist a proper budgeted coloring for $G$ then there exist a proper budgeted coloring where each $P_i$ gets one of the colors from $L_i$. We denote a coloring $\varphi$, a {\sc restricted proper budgeted coloring} with respect to $\Pe$ if for all $P_i\in\Pe$ any two vertices $u,v\in P_i$, $\varphi(u)=\varphi(v)$.

\begin{lemma}\label{lemma:color_VC}
	If there exists a restricted proper budgeted coloring of $G$ with respect to $\Pe$, then there exists a restricted proper budgeted coloring where the vertices in each partition $P_i$ is colored with one of the colors in $L_i$.  
\end{lemma}
\begin{figure}[!htb]
	\centering
	\includegraphics[width= 8 cm]{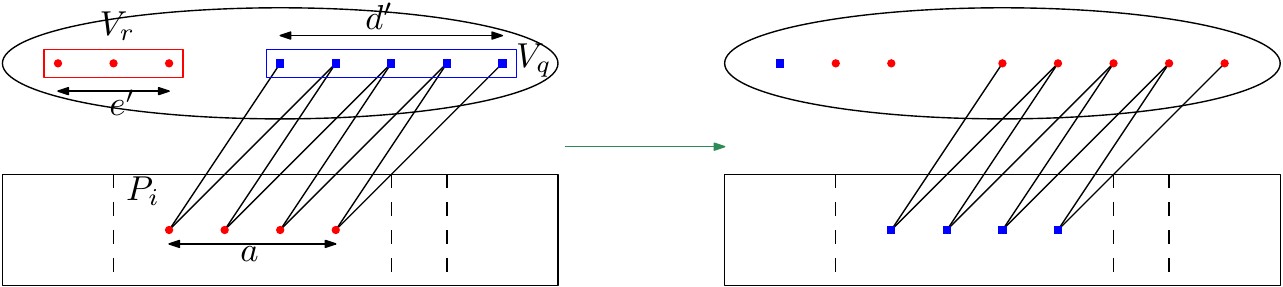}
	\caption{Illustration of Lemma \ref{lemma:color_VC}, color $q$ and $r$ are denoted by {\color{blue}{\rule{5 pt}{5 pt}}} and  {\color{red}$\bullet$} respectively}
	\label{fig:vc} 
\end{figure}
\begin{proof}
	For the sake of contradiction assume that there does not exist a restricted proper budgeted coloring such that all the vertices in each partition $P_i$ are colored with one of the colors in $L_i$. Let $\varphi^*$ be the coloring that maximizes the number of partitions $P_i$ that are colored with one of the colors in $L_i$. Observe that if for a partition $|L_i|\leq \ell$ then $L_i=F_i$ thus must be colored with one of the colors in $L_i$ in any restricted budgeted proper coloring.	
	Let $P_i\in \Pe$ be the partition that is colored with a color $r\notin L_i$. As $L_i$ contains at least $\ell$ colors there exist a color say $q\in L_i$ which is not used to color any of the partitions in $\Pe$. Observe that $b_q \leq b_r$ (see Figure \ref{fig:vc}).
	Let $|P_i|=a$, $b_q= a+d$ and $b_r= a+e$ where $0 \leq d\leq e$. Next, we construct a coloring $\varphi'$ from $\varphi^*$ as follows. If $d=e$ then exchange the color $r$ and $q$ to construct $\varphi'$ but this will contradict the maximality of  $\varphi^*$. Thus assume that $d<e$. Let $V_r$ be the set of vertices in $V\setminus S$ which are colored with color $r$. Observe that there are no edge between $P_i$ and $V_r$. Let $|V_r|=e'\leq e$.
	Let $V_q$ be the set of vertices in $V\setminus S$ which are colored with color $q$. Suppose $|V_q|=d'\leq a+d$. 
	
	We construct $\varphi'$ from $\varphi^*$ as follows. We color the vertices of $P_i$ with the color $q$ and the vertices of $V_q$ with $r$ (observe that $b_q<b_r$). We color $V_r$ with rest of the colors of $q$ and $r$. In order to prove that there are sufficient colors left to color $V_r$ let us observe the following. 
	\begin{align*}
		|V_r|&=e'\leq e\leq e+a+d-d' \text{ (as $d'\leq a+d$)}\\
		&= d+(e+a-d')=(b_q-a)+(b_r-d')
	\end{align*}
	We color the rest of the vertices with the same color as of $\varphi^*$. Observe that $\varphi'$ is a restricted proper budgeted coloring and contradicts the maximality of $\varphi^*$ and the claim holds.
\end{proof}

For each part $P_i$ in the partition, we first find a set of the least 
budgeted $\ell$ colors each of which is of size at least $|P_i|$ (if there aren't $\ell$ different colors satisfying the budget constraint, we pick all those that satisfy; if there aren't any color satisfying the budget constraint, we abandon this partition and move on). Then we try all possible colorings of $S$ coloring each $P_i$ with any of the colors we have found for $P_i$ making sure that no pairs of $P_i$'s get the same color. 
As soon as we fix the coloring of each $P_i$, we are left with the independent set vertices. Since each vertex in the independent set is itself a singleton cluster, we have to solve the \ebcp. By Lemma~\ref{lemma:cvd+colorscorrectness}, it is known that the problem can be solved in polynomial time.

As the number of partitions $\mathcal{P}$ of $S$ is $\mathcal{O}(k^k)$, and as the number of colorings tried for each partition is $\mathcal{O}^*(k^k)$, we have the following theorem.
\begin{theorem} \label{TH:VC}
	\bcp parameterized by vertex cover can be solved in $\mathcal{O}^{*}(2^{2k \log k})$ time.
\end{theorem}

It is known that \ecp does not admit a polynomial kernel when parameterized by vertex cover and the number of colors unless \textsf{NP $\subseteq$ coNP/poly}~\cite{gomes2019structural}. Hence the same result is true for \ecp
when parameterized by vertex cover alone. As \bcp is a generalization of \ecp we have the following corollary.
\begin{corollary}
	\bcp does not admit a polynomial kernel when parameterized by vertex cover unless \textsf{NP $\subseteq$ coNP/poly}.
\end{corollary}

%% file: exact.tex

\section{Exact Algorithm for Budgeted Coloring}\label{sec:exact}
In this section, we present an exact exponential time algorithm for the \bcp based on dynamic programming. Let $G[X]$ denote the subgraph induced by $X$, for $X \subseteq V(G)$. 
A $\mathcal{O}^*(3^n 2^c)$ algorithm is easy by computing a table $T$ of size $2^{n} \times 2^{c}$ whose entries are as follows. For a subset $X$ of vertices, and a subset $C$ of colors, the entry $T[X,C]$ will contain $1$, if $G[X]$ can be colored with the colors in $C$ along with their respective budgets and $0$ otherwise. Now the recurrence relation for this problem is as follows:
$T[S,C]=\underset{I \in G[S], c \in C}{\bigvee}T[S \setminus I,C \setminus \{c\}]$
where $I$ is an independent set in $G[S]$ of size at most $b_c$. 
Trivially, when both $S$ and $C$ is $\emptyset$, $T[S,C]=0$. Our algorithm runs on all possible subsets $S$ of $V(G)$, and for every subset $S$, it runs on all its independent sets, and all possible colors, which is bounded by $c 2^{|S|}$. For every $S$, we have table entries for all possible subset $C$ of colors which is bounded by $2^{c}$. 
Hence, the total running time of this algorithm can be easily seen to be $ \sum_{i=1}^{n} \binom{n}{i} 2^{i} \cdot 2^{c}\cdot c$ which is $=\mathcal{O}(3^{n}\cdot 2^{c}\cdot c)$.

In what follows we improve the runtime to $\mathcal{O}^*(2^n)$ using the principle of inclusion-exclusion, essentially generalizing the known algorithm for the proper $c$-coloring problem~\cite{FominKratschbook}
to show the following.
\begin{theorem}\label{exact}
	Budgeted $c$-Coloring problem can be solved in $\mathcal{O}^*(c 2^n)$ time.
\end{theorem}
\begin{proof}
	We give an algorithm with the claimed bound for the more general {\sc Budgeted Set Cover} problem defined as follows.
	
	\noindent
	{\bf Input:} A ground set $U$ of size $n$, and a family $\mathcal {F}$ of $m$ subsets of $U$, and a set $b_1, b_2, \ldots ,b_c$ of integers for $c \geq 1$. \\
	{\bf Question:} Are there subsets $F_1, F_2, \ldots ,F_c$ in $\mathcal {F}$ such that $|F_i| \leq b_i$ and $\bigcup _{i=1}^c F_i = U$?
	
	In this problem, we can even assume that the family $\mathcal {F}$ is given implicitly in the sense, given a subset $S$ of $U$, one can test in time polynomial in $|S|$ whether or not $S$ is in $\mathcal {F}$. We call a subfamily that witnesses a solution to our problem as a budgeted $c$-cover of $\mathcal {F}$.
	
	Now by treating the family $\mathcal {F}$ as the family of independent sets of the input graph of the budgeted $c$-coloring, the claim in the theorem follows.
	To design an algorithm for {\sc Budgeted Set Cover} using inclusion-exclusion, we define our object of interest. The object of interest is simply a subfamily 
	$\mathcal{F'} = \{F_1, F_2, \ldots ,F_c \}$ of $\mathcal{F}$ such that $|F_i| \leq b_i$; observe that in our definition for $i\neq j$ $F_i$ and $F_j$ can be the same. 
	The subfamily $\mathcal{F'}$ satisfies property $P(u)$ if $u \in \bigcup _{F \in \mathcal{F'}} F.$
	
	Now it follows that the number of budgeted $c$-covers of $\mathcal{F}$ is simply the number of objects that satisfy $P(u)$ for every $u \in U$ (as our objects are ordered, the same budgeted $c$-cover will be counted a fixed number ($c!)$ times, and hence the final number has to be divided by $c!$ which we ignore hereafter).
	
	For a subset $W \subseteq U$, let $F(W, b_i)$ be the family of sets in $\mathcal{F}$ that avoids $W$ (i.e doesn't have any element of $W$) and are of size at most $b_i$, and let $f(W, b_i)$ be the number of such sets.
	Then the number of objects that do not satisfy property $P(u)$ for any element $u$ of $W$ is simply $\Pi _{i=1}^c f(W, b_i)$. This is simply because there are $f(W,b_i)$ choices for the $i$-th set of our object.
	
	Hence, by the principle of inclusion-exclusion (see Theorem $4.7$ of \cite{FominKratschbook}), we have that the number of budgeted $c$-covers of $\mathcal{F}$ is given by
	\begin{equation}\label{eq:exact}
	   \sum_{W \subseteq U} (-1)^{|W|} \Pi _ {i=1}^c f(W,b_i) 
	\end{equation}
	
	For a fixed $W$ and $b_i$, $f(W,b_i)$ is the number of subsets of $U\setminus W$ of size at most $b_i$ that are in $\mathcal{F}$ and hence can be computed in ${n-|W| \choose b_i} {|b_i|}^{\mathcal{O}(1)}$ which can result in an $\mathcal{O}^*(3^n)$ time for all $W$.
	In the following claim, we show that we can pre-compute and store $f(W,b_i)$ for each $W$ and $b_i$ using dynamic programming and compute them all in $\mathcal{O}^*(2^n)$ time. 
	\begin{claim}\label{claim:exact} 
	The quantity $f(W,b_i)$, for all $W$ and $b_i, i=1$ to $c$ can be computed in $\mathcal{O}^*(2^n)$ time.
	\end{claim}
	\begin{proof}
	Let $u_1, u_2, \ldots u_n$ be the set of elements in $U$. Let us define $g(W, b_i, j)$ be the number of sets of size at most $b_i$ in $\mathcal{F}$ that avoid $W$, but (otherwise) contain all elements of
	$u_{j+1}, u_{j+2}, \ldots u_n$. Then $g(W, b_i, 0) = 1$ if $U \setminus W$ is in $\mathcal{F}$ and is of size at most $b_i$ and $0$ otherwise.
	
	Also $g(W, b_i, n) = f(W,b_i) $ is the number of sets of $\mathcal{F}$ of size at most $b_i$ that avoid $W$, that we want to compute.
	
	And for a $W$ and $b_i$, we can compute $g(W, b_i, j)$ once we have $g(W, b_i, j-1)$ for all $W$ and $b_i$ as follows.
	
	\noindent
	$ g(W, b_i, j) = g(W, b_i, j-1)$ if $u_j \in W$, and \\
	$ g(W, b_i, j) = g(W, b_i, j-1) + g(W \cup \{u_j\}, b_i, j-1)$
	
	It follows that $g(W, b_i, n)$ for all $W$ and $i$ can be computed in $\mathcal{O}(2^n c n)$ time.
	This completes the proof of the claim.
\end{proof}

	The theorem follows from the Equation~\ref{eq:exact} and the Claim.
\end{proof}

%% file: Conclusion.tex
\section{Conclusions and Open Problems}\label{sec:conclusions}
In this paper, we have introduced the \bcp and obtained hardness results, polynomial time algorithms, \textsf{FPT} algorithms and kernelization results. 
There are a number of open problems.

\begin{itemize}
\item What is the complexity of the \bcp on trees when $c \geq 3$? The related \ecp and \textsc{Bounded Coloring Problems} are polynomial time solvable on trees~\cite{chen1994equitable,baker1996mutual}. 
On the other hand there is a generalization of \bcp to the list coloring vairant called  \textsc{Number List Coloring Problem} that is known to be $\textsf{W}[1]$-hard on forests parameterized by the total number of colors that appeares on the lists\cite{fellows2011complexity}.

\item In Section~\ref{sec:param_CVD}, we showed that \bcp when parameterized by the distance to a cluster graph and the number of colors or the number of clusters, is $\textsf{FPT}$. The parameterized complexity of the problem parameterized just by the distance to a cluster graph is an interesting open problem.
\end{itemize}